\documentclass[fleqn,10pt]{wlscirep}

\usepackage{calc}
\usepackage{textcomp}
\usepackage{graphicx}
\usepackage{hyperref}
\usepackage[caption=false]{subfig}
\usepackage{braket}

\usepackage{amsmath,amsthm,amsfonts,amssymb,amscd}

\usepackage[utf8]{inputenc} 

\usepackage{tabularx}

\usepackage{multirow}

\usepackage{longtable}

\newtheorem{prop}{Proposition}

\usepackage[normalem]{ulem}

\makeatother

\title{Optimal strategy to certify quantum nonlocality}

\author[1,2,*]{S. G\'omez}
\author[3]{D. Uzc\'ategui}
\author[1,2]{I. Machuca}
\author[1,2]{E. S. Gómez}
\author[1,2]{S. P. Walborn}
\author[1,2]{G. Lima}
\author[3]{D. Goyeneche}
\affil[1]{Departamento de Física, Universidad de Concepci\'on, 160-C Concepci\'on, Chile.}
\affil[2]{ANID—Millennium Science Initiative Program—Millennium Institute for Research in Optics, Universidad de Concepci\'on, Concepci\'on 160-C, Chile}
\affil[3]{Departamento de F\'{i}sica, Facultad de Ciencias B\'{a}sicas, Universidad de Antofagasta, Casilla 170, Antofagasta, Chile.}

\affil[*]{santgomez@udec.cl}

\begin{abstract}
 Certification of quantum nonlocality plays a central role in practical applications like device-independent quantum cryptography and random number generation protocols. These applications entail the challenging problem of certifying quantum nonlocality, something that is hard to achieve when the target quantum state is only weakly entangled, or when the source of errors is high, e.g. when photons propagate through the atmosphere or a long optical fiber. Here we introduce a technique to find a Bell inequality with the largest possible gap between the quantum prediction and the classical local hidden variable limit for a given set of  measurement frequencies. Our method represent an efficient strategy to certify quantum nonlocal correlations from experimental data without requiring extra measurements, in the sense that there is no Bell inequality with a larger gap than the one provided. Furthermore, we also reduce the photodetector efficiency required to close the detection loophole. We illustrate our technique by improving the detection of quantum nonlocality from experimental data obtained with weakly entangled photons.
\end{abstract}
\begin{document}

\flushbottom
\maketitle
\section{Introduction} 
Quantum nonlocality plays a fundamental role in flourishing quantum technologies, such as device \cite{MY98,ABGPS07} and semi-device \cite{PB11} independent quantum cryptography, and genuine random number generation \cite{PAM10}, as well in fundamental aspects of quantum physics. In these applications, the certification of nonlocality is usually a required step. In the ideal case, for any set of nonlocal correlations, there exists a Bell inequality that is violated \cite{BCPSW14}. However,  certification of nonlocality can be hard to achieve in practice due to experimental errors. This is especially true when the optimal quantum state, i.e. the state producing the maximal violation of a Bell inequality, is weakly entangled \cite{AMP12}. This problem plays a relevant role in tight Bell inequalities, that might be maximally violated by partially entangled quantum states \cite{CGLMP02,VW11}.  These inequalities are particularly useful as they are known to maximize the randomness that can be certified in a Bell scenario. For instance, a recent experiment using near-ideal two-qubit states failed to certify quantum nonlocality for a weakly entangled quantum systems \cite{GMMGCJAL19}.  

Bell inequalities can be used to test that a set of correlations cannot be described by a local hidden variable (LHV) model. On the one hand, bipartite Bell inequalities can have a large ratio between the achievable quantum and LHV limits, equal to  $\sqrt{n}/\log{n}$, for $n$ settings and $n$ outputs in $n$ dimensional Hilbert spaces 
\cite{JP11}. However, when testing quantum nonlocality, one is particularly interested in a state and measurements that maximize the violation of a given Bell inequality. From the experimental perspective, a larger theoretical violation increases the chance to certify quantum nonlocality in the laboratory. Nonetheless, sometimes experiments are not conclusive to certify nonlocality. When failing to test nonlocality in the laboratory, one can  choose another Bell inequality with a larger gap between the LHV and quantum values, thus increasing the chances for success. However, the cost of this option is to implement a new experiment, as the optimal settings of the new Bell inequality most likely differ from the original one. This procedure consumes additional time and resources. 

Thus, a fundamental question arises: \emph{Can we certify quantum nonlocality from experimental data that previously failed to violate a target Bell inequality?} In this work, we find necessary and sufficient conditions to provide a conclusive answer to this question, for any bipartite scenario. In addition, we present an optimization method that finds a Bell inequality that maximizes the chances to detect quantum nonlocality, among the entire set of inequalities of a given scenario. This result has a wide range of practical applications including communication complexity problems, where the advantage in communication is an increasing function of the quantum-LHV value gap \cite{BZPZ04,TZB20}.

\section{Method}\label{sec:optimization}
In this section, we introduce a method that provides the largest possible gap between the quantum and LHV predictions regarding a fixed set of experimental data. In particular, this procedure allows us to determine whether a set of experimental data is genuinely nonlocal or not, i.e. whether there is a Bell inequality that can certify quantum nonlocality from a given data. Our method represents an efficient certification of nonlocal correlations, which can be applied to experimental data without requiring extra measurements. That is, it produces a Bell inequality that maximizes the chances to detect quantum nonlocality from a given set of statistical data among the entire set of Bell inequalities.

A bipartite Bell inequality \cite{BCPSW14} has the form
\begin{equation}\label{Bellineq}
\sum_{x,y=0}^{m-1}\sum_{a,b=0}^{d-1} s^{ab}_{xy}\,p(a,b|x,y)\leq\mathcal{C}(s),
\end{equation}
where $p(a,b|x,y)$ is the probability of obtaining outcomes $a,b\in\{0,\dots,d-1\}$ when inputs $x,y\in\{0,\dots,m-1\}$ are chosen by two observers Alice and Bob, respectively.  Here, $\mathcal{C}(s)$ denotes the maximal value of the left hand side in (\ref{Bellineq}) that can be achieved by considering \emph{local hidden variable} (LHV) theories, whereas quantum mechanics might predict a violation of this inequality  \cite{B64}. Without loss of generality, we can restrict out attention to parameters within the set $-1\leq s^{ab}_{xy}\leq1$, for every $a,b=0,\dots,d-1$ and $x,y=0,\dots,m-1$. 

Quantum joint probability distributions satisfy the no-signaling principle, implying that information cannot be instantaneously transmitted between distant parties. In particular, the outcome of one party cannot reveal information about the input of the other. That is,
\begin{equation}\label{ns1}
\sum_{b=0}^{d-1}p(a,b|x,y)=\sum_{b=0}^{d-1}p(a,b|x,y')=:p_A(a|x),
\end{equation}
and
\begin{equation}\label{ns2}
\sum_{a=0}^{d-1}p(a,b|x,y)=\sum_{a=0}^{d-1}p(a,b|x',y)=:p_B(b|y),
\end{equation}
for every $x\neq x'$ and $y\neq y'$, where $p_A(a|x)$ and $p_B(b|y)$ are the marginal probability distributions associated to Alice and Bob, respectively.\\
\indent Let us now consider a set of relative frequencies  $f(a,b|x,y)$ of occurrence for outcomes $a,b$ when $x,y$ is measured by Alice and Bob, respectively, obtained from experimental data. The no-signaling constraints (\ref{ns1}) and (\ref{ns2}) do not occur due to errors but they can be recovered by minimizing the Kullback-Leible divengence \cite{SRZBC18}:
\begin{equation}\label{KLdivergence}
D_{KL}(\vec{f}||\vec{P})=\sum_{a,b,x,y}f(x,y)f(a,b|x,y)\log_2\left[\frac{f(a,b|x,y)}{P(a,b|x,y)}\right],
\end{equation}
where $f(x,y)$ is the relative frequency of implementing a measurement $x$ by Alice and $y$ by Bob, and $p(a,b|x,y)$ the optimization variables, consisting of a joint probability distribution within the framework of quantum mechanics. The minimization procedure (\ref{KLdivergence}) is equivalent to maximizing the likelihood of producing the observed frequency $p(a,b|x,y)$, see Appendix D1 in \cite{SRZBC18}.

The experimental prediction of a Bell inequality (\ref{Bellineq}) defined by coefficients $s^{ab}_{xy}$ is given by  
\begin{equation}\label{Qexp}
\mathcal{Q}=\sum_{x,y=0}^{m-1}\sum_{a,b=0}^{d-1} s^{ab}_{xy}\,p(a,b|x,y),
\end{equation}
having associated an error propagation $\Delta\mathcal{Q}$, see section A, supplementary information, for a detailed treatment of errors. An experimentally obtained probability distribution $p(a,b|x,y)$, associated to errors $\Delta p(a,b|x,y)$, is certainty nonlocal if $\mathcal{Q}-\mathcal{C}>\Delta\mathcal{Q}$, for a given Bell inequality. However, sometimes quantum nonlocality cannot be revealed due to the amount of errors,  especially when a weakly entangled quantum state produces the maximal violation of the inequality. Under such situation, the method introduced provides a new Bell inequality that increases the chances to prove quantum nonlocality for a given set of probability distributions $p(a,b|x,y)$, associated to experimental errors $\Delta p(a,b|x,y)$. The method consists in solving the following nonlinear problem:
\begin{equation}\label{max}
R=\max_{s}\frac{\mathcal{Q}(s)-\Delta \mathcal{Q}(s)+dm}{\mathcal{C}(s)+dm},
\end{equation}
for a fixed set of statistical data, where the optimization is implemented over all coefficients $s^{ab}_{xy}$ defining a Bell inequality (\ref{Bellineq}). The shifting factor $dm$ introduced in (\ref{max}) avoids divergence of the function $R$, otherwise, the output inequality would be any such that the LHV value $\mathcal{C}(s)$ vanishes. Optimization (\ref{max}) is typically hard to implement, due to the presence of a large amount of local maximum values. To maximize the chances to find the global maximum, we studied four different optimization techniques: a gradient-descendent method \cite{C44} and three direct search methods: Nelder–Mead \cite{NM65}, differential evolution \cite{SP97}, and simulated annealing \cite{KGV83}. The best optimization strategy consists in taking a large number of random seeds and apply a gradient-descendent method, according to our experience. All our numerical simulations were implemented in Python. A code to solve optimization (\ref{max}) is available at GitHub \cite{code}. 

Let us summarize the method as follows:
\begin{prop}\label{prop1}
A Bell inequality having LHV value $\mathcal{C}(s)$, for which a quantum value $\mathcal{Q}(s)$ is achieved with errors $\Delta\mathcal{Q}(s)$ is nonlocal if and only if $R>1$ in Eq.(\ref{max}).
\end{prop}
\begin{proof}
Let $s^{ab}_{xy}$ be the parameters of the Bell inequality producing the maximum value $R$ in (\ref{max}). Therefore, it is simple to show that $\mathcal{Q}(s)-\Delta\mathcal{Q}(s)-\mathcal{C}(s)=(R-1)(\mathcal{C}(s)+dm)$. Therefore, the statistical data is nonlocal, i.e. $\mathcal{Q}(s)-\Delta\mathcal{Q}(s)-\mathcal{C}(s)>0$, if and only if $R>1$. Given that $R$ takes the maximal possible value among all Bell inequalities of the scenario, if $R\leq1$ then there is no Bell inequality that can detect quantum nonlocality for the given statistical data. 
\end{proof}
Here, genuinely nonlocal means that there exists a Bell inequality (\ref{Bellineq}), associated to coefficients $s^{ab}_{xy}$, such that the amount of experimental errors do not overpass the gap between the quantum violation $\mathcal{Q}(s)$ and the LHV prediction $\mathcal{C}(s)$. Therefore, $R>1$ is equivalent to saying that there is a way to experimentally certify quantum nonlocality for a given set of experimental data. On the other hand, if $R\leq1$ then there is no way to decide whether the physical system is prepared in a nonlocal quantum state or not. 

Note that the gap of a Bell inequality $\mathcal{Q}(s)-\Delta\mathcal{Q}(s)-\mathcal{C}(s)$ can be artificially enlarged by a multiplicative factor $\kappa>1$ by considering a Bell inequality having a rescaled set of  coefficients $\tilde{s}^{ab}_{xy}=\kappa s^{ab}_{xy}$. In order to avoid this scaling problem, it is convenient to refer to the \emph{relative gap} of a Bell inequality, given by the gap of a Bell inequality for which $\mathcal{C}(s)=1$. Assuming this convention, the global maximum value of $R$ in (\ref{max}) implies the largest possible relative gap, as we show below.
\begin{prop}
The Bell inequality associated to the global maximum of the function $R$, introduced in (\ref{max}), produces the largest possible relative gap between the LHV and quantum values, among all Bell inequalities of a given scenario. 
\end{prop}
\begin{proof}
Without loss of generality, we can restrict our attention to Bell inequalities for which $\mathcal{C}(s)=1$. Indeed, this can always be done by considering the following rescaling of a given Bell inequality: $\mathcal{Q}'(s)=\mathcal{Q}(s)/\mathcal{C}(s)$,  $\Delta\mathcal{Q}'(s)=\Delta\mathcal{Q}(s)/\mathcal{C}(s)$ and $\mathcal{C}'(s)=1$. Here, we used the fact that error propagation is linear as a function of the coefficients $s^{ab}_{xy}$ defining the Bell inequality, see section A of supplementary information.  After the rescaling, optimization of $R$ is equivalent to maximize $\mathcal{Q}'(s)-\Delta\mathcal{Q}'(s)$ along all Bell inequalities satisfying $\mathcal{C}'(s)=1$, according to (\ref{max}). This is equivalent to maximizing the gap $\mathcal{Q}'(s)-\Delta\mathcal{Q}'(s)-1=\mathcal{Q}'(s)-\Delta\mathcal{Q}'(s)-\mathcal{C}'(s)$.
\end{proof}

In the next section, we apply our optimization method to the so-called \emph{tilted Bell inequality} \cite{AMP12}. This family of inequalities was used to demonstrate that almost two bits of randomness can be extracted from a quantum system prepared in a weakly entangled state. This property makes this inequality an ideal candidate to test our method due to the hardness to certify nonlocality in such case.

\section{Tilted Bell inequality}\label{sec:tilted} Certification of quantum nonlocality from experimental data is a challenging task in general. In a recent work,  a genuine random number generation protocol based on quantum nonlocality was experimentally demonstrated \cite{GMMGCJAL19}. Here, authors certified randomness from a physical system prepared in quantum states with concurrences $C=0.986$, $C=0.835$ and $C=0.582$. However, certification failed for $C=0.193$ and $C=0.375$. The inequality in question is the tilted Bell inequality \cite{AMP12}: 
\begin{equation}\label{belltilted}
\alpha \left[ p_{A}(0|0)-p_{A}(1|0) \right]+ \sum_{x,y=0}^{1}\sum_{a,b=0}^{1} (-1)^{xy}\left[ p(a=b|xy)-p(a \neq b|xy)\right] \leq \mathcal{C}_{\alpha},
\end{equation}
where $p_A(a|x)$ is the marginal probability distribution for Alice,  $\mathcal{C}_{\alpha}=\alpha+2$ and $\mathcal{Q}_{\alpha}=\sqrt{8+2\alpha^2}$, $\alpha\in[0,2]$. The  quantum state producing the maximal violation in (\ref{belltilted}) is given by $|\psi(\theta)\rangle=\cos(\theta)|00\rangle+\sin(\theta)|11\rangle$,
where $\theta=\frac{1}{2}\arcsin \left( \sqrt{\frac{1-(\frac{\alpha}{2})^{2}}{1+(\frac{\alpha}{2})^{2}}}\right) $. Optimal settings are provided by eigenvector bases of the following observables:
$A_{0}=\sigma_{z}$, $A_{1}=\sigma_{x}$, $B_{0}= \cos(\mu)\,\sigma_{z} + \sin(\mu)\,\sigma_{x}$ and $B_{1} = \cos(\mu)\,\sigma_{z} - \sin(\mu)\,\sigma_{x}$,
where $\mu=\arctan \left( \sqrt{\frac{1-(\frac{\alpha}{2})^{2}}{1+(\frac{\alpha}{2})^{2}}}\right)$.

\begin{figure}[htp]
\centering
\includegraphics[width=0.55\textwidth]{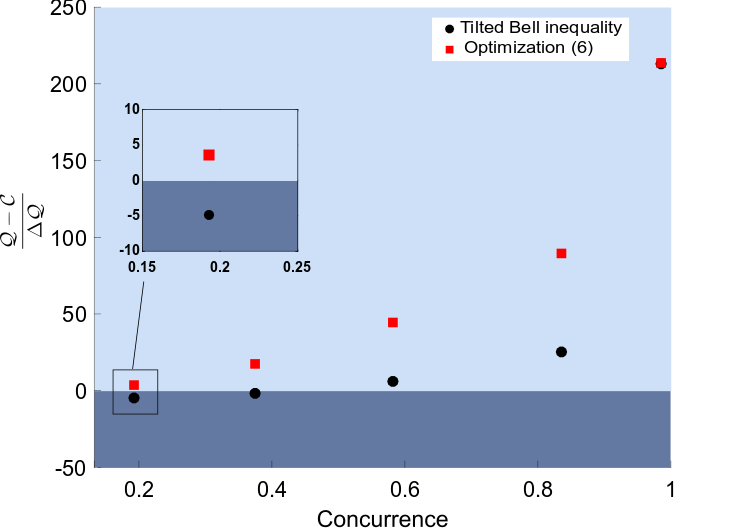}
\caption{The standard deviation number (SDN) as a function of concurrence. For each value of concurrence, the optimization procedure (\ref{max}) provides a  Bell inequality that increases the number of standard deviations of the quantum-LHV value gap. SDN is calculated for two cases: the original tilted Bell inequality [black circles] and an inequality arising from optimization (\ref{max}) [red squares]. In both cases, we consider experimental data, where quantum nonlocality can be certified in the light blue region. For concurrences $C=0.375$ and $C=0.193$, there is no quantum violation of the tilted Bell inequality, whereas inequalities arising from optimization (\ref{max}) produce a violation in all the cases.}
\label{Fig1}
\end{figure}
\newpage
We tested optimization (\ref{max}) from the photonic Bell inequality experiment described in section B of supplementary information (see Fig.S1). The experiment consisted of a high-quality source of polarization-entangled Bell states of the form $\ket{\psi(\theta)}$.  Entanglement of this optimal quantum state can be characterized by its concurrence, given by $C=\sin(2 \theta)$. By using optimization (\ref{max}), we improve the experimental violation of the tilted Bell inequality for high concurrences and, more importantly, we successfully demonstrate quantum nonlocality for low values of concurrence, i.e. $C=0.375$ and $C=0.193$, something that failed to be proven  when considering the tilted Bell inequality \eqref{belltilted} \emph{from the same statistical data}. In Fig. \ref{Fig1} we show the number of standard deviations of the quantum-LHV gap for the tilted Bell inequality and the one resulting from our optimization procedure (\ref{max}) (more details in Table \ref{table:one}).\\

This result can lead to interesting practical applications. For instance, note that the amount of bits of randomness that can be extracted from the protocol \cite{GMMGCJAL19} are identical to the number of bits that can be extracted by considering our optimized Bell inequality. This is so because both schemes consider the same experimental data. 

\begin{table}[h!]
\centering
   \caption{Summary of results. The concurrence was obtained from quantum state tomography. $\mathcal{Q}_{\alpha}$  refers to the experimental values of the Tilted Bell inequality. $\mathcal{C}_{\alpha}$ is the LHV bound for each $\alpha$. After implementing the method mentioned in section II, the results obtained are presented in the values $\mathcal{Q}$ and $\mathcal{C}$.}
\resizebox{0.8\textwidth}{!}{\begin{minipage}
{\textwidth} \begin{center}
        \begin{tabular}{ c  p {1 cm}  p {1 cm}   p {1 cm} p {1 cm} c }
        \hline 
        Concurrence & $0.193$ & $0.375$ & $0.582$  &  $0.835$ & $0.986$ 
        \\
        \hline \hline
        $\mathcal{Q}_{\alpha}$   & $3.890$ & $3.686$ & $3.418$ & $3.108$ & $2.819$  \\
        $(\pm) \Delta \mathcal{Q}_{\alpha}$ & $0.006$ & $0.008$ & $0.008$ & $0.007$ & $0.004$\\ 
        $\mathcal{C}_{\alpha}$ & $3.921$ & $3.702$ & $3.372$ & $2.937$ & $2.002$ \\ 
        $\frac{\mathcal{Q}_{\alpha}-\mathcal{C}_{\alpha}}{\Delta \mathcal{Q}_{\alpha}}$ & $-4.85$ & $-1.93$ & $5.83$ & $25.44$ & $213.06$ 
        \\
        \\
        $\mathcal{Q}$   & $1.436$  & $0.858$ & $2.308$ & $2.290$ & $1.819$ \\
        $(\pm) \Delta \mathcal{Q}$ & $0.001$ & $0.006$ & $0.007$ & $0.006$ & $0.004$  \\ 
        $\mathcal{C}$ & $1.429$ & $0.745$ & $1.994$ & $0.745$ & $1.429$ \\ 
        $\frac{\mathcal{Q}-\mathcal{C}}{\Delta \mathcal{Q}}$ &  $3.54$ & $17.75$ & $44.41$ & $89.60$ & $213.57$         \\    
        \\
        \hline \hline
       
        \end{tabular} 
        \end{center}
	\end{minipage}}
	\label{table:one}
\end{table}
    
\section{Closing the detection loophole}\label{sec:loophole}
In this section, we show that optimization procedure (\ref{max}), apart from increasing the quantum-LHV value gap, also reduces the detection efficiency required to close the detection loophole, for a fixed set of statistical data. 

Suppose that Alice and Bob have detector efficiencies $\eta_A$ and $\eta_B$, respectively. The minimal efficiencies required to violate a Bell inequality are given by the following procedure \cite{Pironio2010}: first, a two outcomes Bell inequality has to be written in a canonical form. Namely, it has to consider only one of its outcomes per party (we choose $a=b=0$), as we associate the other outcomes ($a=b=1$) to the cases where detectors do not fire correctly. It is simple to show that any bipartite Bell inequality (\ref{Bellineq}) with $m$ settings per side and two outcomes can be written as: 
\begin{equation}\label{effi1a}
\sum_{x,y=0}^{m-1} \tilde{s}_{xy}^{00} p(0,0|x,y)+\sum_{x=0}^{m-1} \tilde{s}_{Ax}^{\hspace{6px}0} p_A(0|x)+\sum_{y=0}^{m-1} \tilde{s}_{By}^{\hspace{6px}0} p_B(0|y)=\mathcal{C}(s),
\end{equation}
where 
\begin{eqnarray}
\tilde{s}^{00}_{xy}&=&\sum_{a,b=0}^1 (-1)^{a+b}s^{ab}_{xy},\nonumber\\
\tilde{s}_{Ax}^{\hspace{6px}0}&=&s_{Ax}^{\hspace{6px}0}-s_{Ax}^{\hspace{6px}1}+\sum_{y=0}^{m-1}\sum_{a=0}^1 (-1)^a s^{a1}_{xy},\nonumber\\
\tilde{s}_{By}^{\hspace{6px}0}&=&s_{By}^{\hspace{6px}0}-s_{By}^{\hspace{6px}1}+\sum_{x=0}^{m-1}\sum_{b=0}^1 (-1)^b s^{1b}_{xy}.\nonumber
\end{eqnarray}
These transformations arises from the identities shown in section C of supplementary information. For instance, for the tilted Bell inequality (\ref{belltilted}), we have the following canonical form:
\begin{equation}\label{canonical_tilted}
(\alpha/2-1) p_A(0|0)- 
 p_B(0|0)+p(0,0|0,0) + p(0,0|0,1) +p(0,0|1,0)-
 p(0,0|1,1)\leq \tilde{\mathcal{C}}_{\alpha},
\end{equation}
where $\tilde{\mathcal{C}}_{\alpha}=(\mathcal{C}_{\alpha}+\alpha-2)/4$.

Second, due to imperfect detectors, probabilities have to be transformed according to the the rule $p(0,0|x,y)\rightarrow \eta_A\eta_B p(0,0|x,y)$, $p(0|x)\rightarrow \eta_A\, p(0|x)$ and $p(0|y)\rightarrow \eta_B\, p(0|y)$, for every setting $x,y$. Therefore, the lower bound for the minimal efficiencies required to detect a quantum violation satisfies:
\begin{equation}\label{effi1b}
\eta_A\eta_B\sum_{x,y=0}^{1} \tilde{s}_{xy}^{00} p(0,0|x,y)+\eta_A\sum_{x=0}^{1} \tilde{s}_{Ax}^{\hspace{6px}0} p_A(0|x)+\eta_B\sum_{y=0}^{1} \tilde{s}_{By}^{\hspace{6px}0} p_B(0|y)=\mathcal{C}(s).
\end{equation}
We observe that optimization (\ref{max}) reduces the detection efficiency required to close the detection loophole, with respect to a target Bell inequality, for a fixed set of probability distributions $p(a,b|x,y)$, $p_A(0|x)$ and $p_B(0|y)$. Here, we show evidence of this fact by improving detection efficiencies for the experimental statistical distributions associated to the tilted Bell inequality (\ref{canonical_tilted}), see Figure \ref{efficiency}. \medskip

\begin{figure}[htp]
\centering
 \includegraphics[scale=0.48]{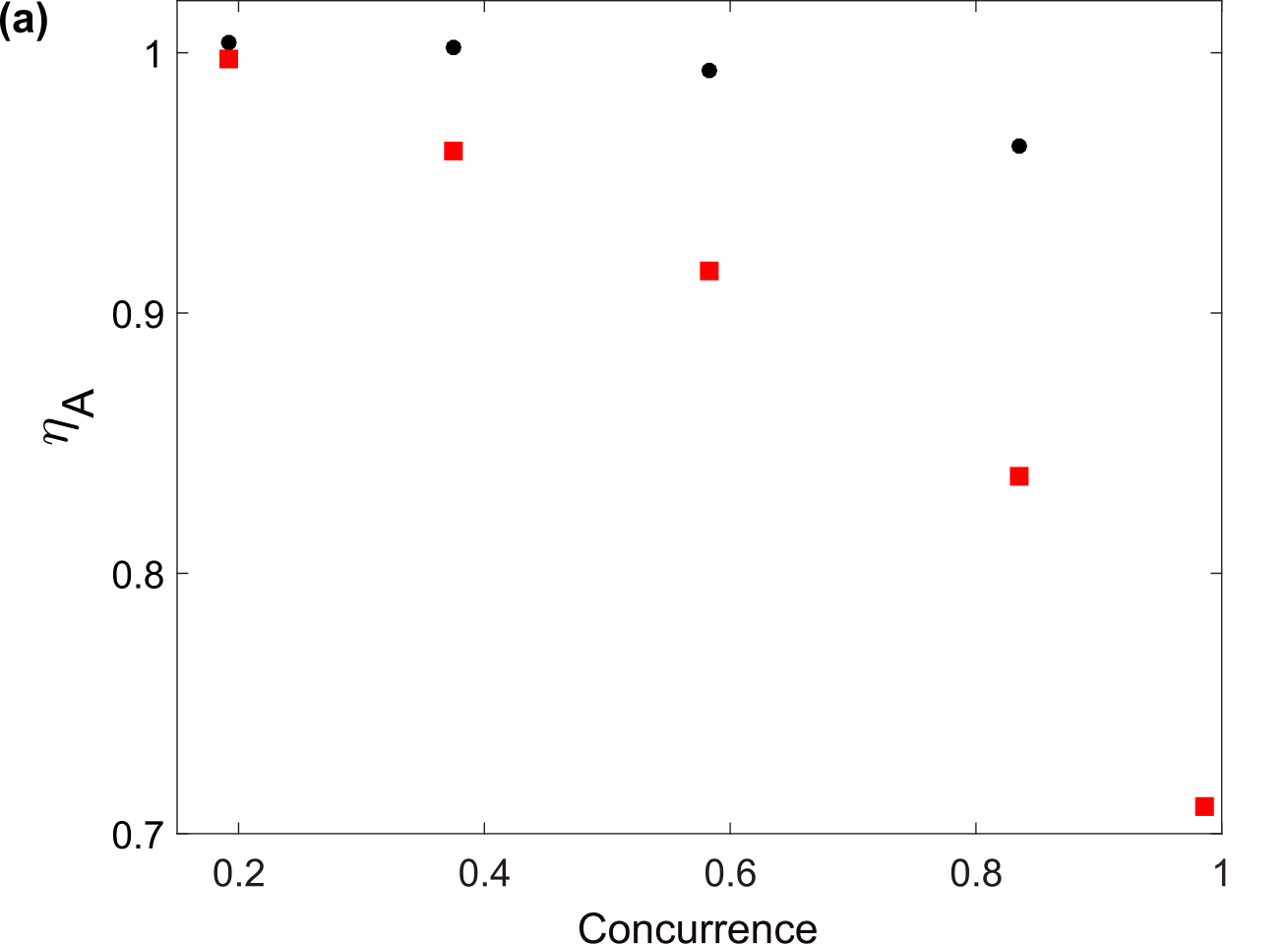}\includegraphics[scale=0.48]{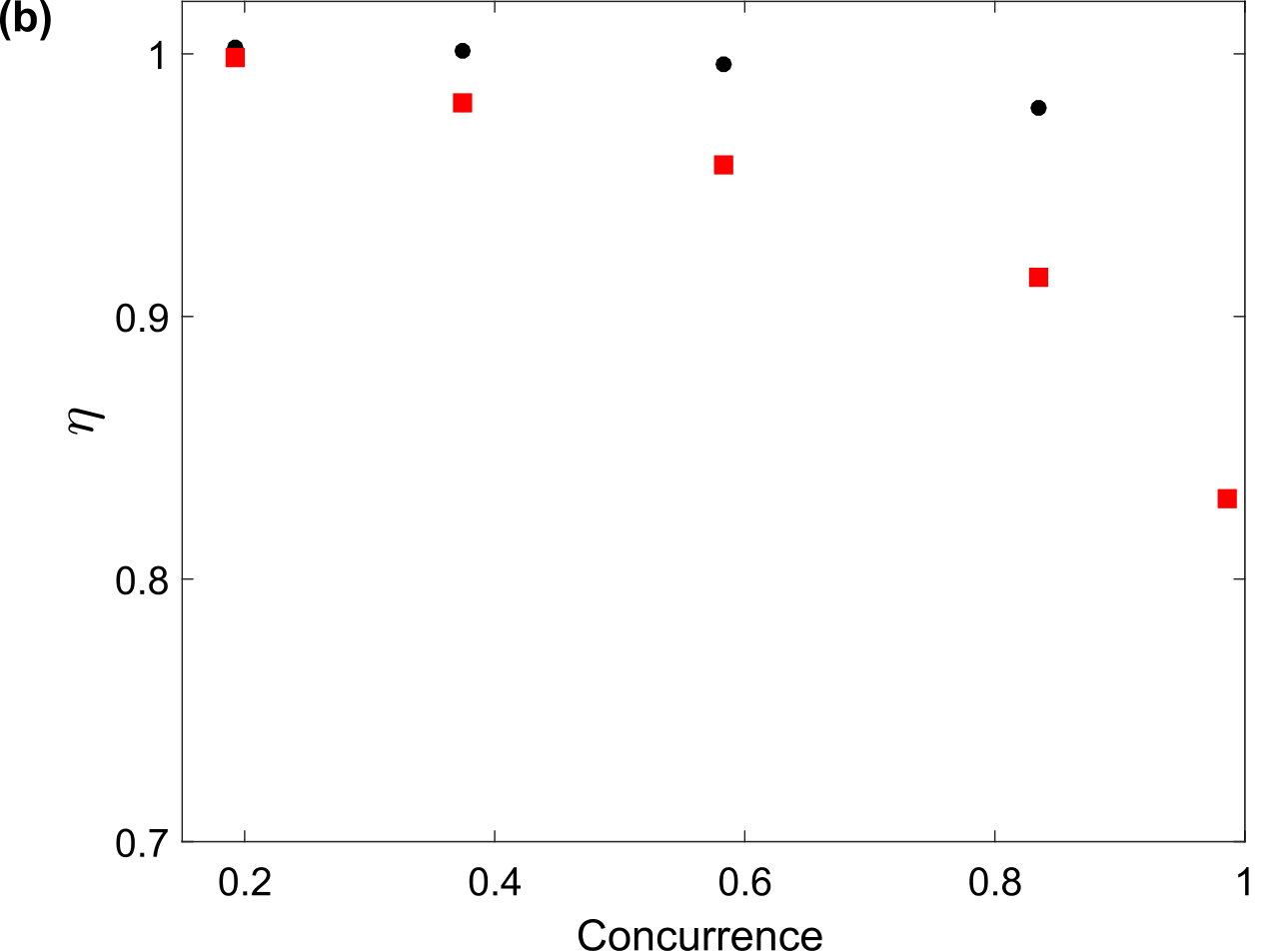}
\caption{(a) Asymmetric case ($\eta_B=1$) and (b) symmetric case ($\eta_A=\eta_B=\eta$). Minimum efficiencies required to close the loophole. We compare those associated to the tilted Bell inequality [black circles] with those inequalities arising from optimization (\ref{max}) [red squares], for the optimal statistical data of the tilted inequality \eqref{canonical_tilted}.}\label{efficiency}
\end{figure}

To summarize this section, we conclude that optimization (\ref{max}), for a given set of quantum statistics, produces a Bell inequality that requires a lower detection efficiency, with respect to a target inequality and the same set of statistical data. This property makes our method an interesting tool to increase the the chances to close the detection loophole from a given experimental data.

\section{Conclusions}
Violation of Bell inequalities is at the heart of quantum physics and defines a cornerstone for a wide range of quantum information protocols with real-world appeal.  We have shown how an ``experiment-inspired" optimization procedure can be applied to the search of Bell inequalities that increase the quantum-LHV gap with respect to a target Bell inequality, for a fixed set of experimental data. Furthermore, the gap provided by the optimization procedure is the largest possible among the entire set of Bell inequalities having LHV equal to 1, something that can be assumed without loss of generality. When nonlocality certification from a given set of experimental data fails, our method provides a ``second chance", without requiring one to perform any additional measurement. Furthermore, our optimization also provides a gain in the minimal detection efficiency required by a fixed statistical set, a crucial ingredient to  maximize the chances to close the detection loophole in experiments.  We illustrated our technique by considering  experimental data associated to the maximal violation of the so-called \emph{tilted Bell inequality}. Here, we considerably increased some previously obtained gaps, a fact that allowed us to certify nonlocality for weakly entangled states, something that was not possible to do with the tilted Bell inequality (see Section \ref{sec:tilted}). Furthermore, we also showed how the detection efficiencies can be decreased in this particular case after implementing our optimization procedure (see Section \ref{sec:loophole}). 

Our technique finds application in device-independent protocols, random number generation, communication complexity and any practical application based in quantum nonlocality.

\section*{Acknowledgements}
This work was supported by the Fondo Nacional de Desarrollo Científico y Tecnológico (FONDECYT) (Grants No. 11180474, No. 1190901, No. 1200266, No. 1200859 and No. 3210359), and the National Agency of Research and Development (ANID) Millennium Science Initiative Program ICN17\_012. DU is supported by Grant FONDECYT Iniciaci\'{o}n number 11180474 and MINEDUC-UA project  ANT1956 (Chile). \vspace{1cm}

\begin{center}
{\Large \textbf{Appendix}}
\end{center}

\appendix

\section{Error propagation}\label{ap1}
\noindent This section provides a general expression for the experimental error obtained by measuring the Bell inequality value. In particular, we show how errors in the photon counting number due to finite statistics propagates to $\Delta \mathcal{Q}$. Consider the following general expression for a Bell inequality: 
\begin{equation}\tag{S1}
 \mathcal{Q} =\sum\limits_{x,y = 0}^{m - 1}  \sum\limits_{a,b = 0}^{d-1} s_{xy}^{ab} p(ab|xy) + \sum\limits_{x = 0}^{m - 1}  \sum\limits_{a = 0}^{d-1} s_{x}^{a} p_{A}(a|x) + \sum\limits_{y = 0}^{m - 1}  \sum\limits_{b = 0}^{d-1} s_{y}^{b} p_{B}(b|y).
 \label{eq 1}
\end{equation}

\noindent Here we include both joint and marginal probability distributions. As is typical, the marginal probabilities are calculated from the join probabilities, and we average over all possible $x$ (or $y$), i.e. $ p(a|x) = \dfrac{1}{m} \sum_{y = 0}^{m-1}\sum_{b = 0}^{d-1} p(ab|xy)$ and $  p(b|y) = \dfrac{1}{m} \sum_{x = 0}^{m-1}\sum_{a = 0}^{d-1} p(ab|xy)$. Replacing these quantities in Eq.\eqref{eq 1} and rewriting $\mathcal{Q}$ in term of the coincidence count $c(ab|xy)$ we get
\begin{equation}\tag{S2}
   \mathcal{Q}=\sum\limits_{x,y = 0}^{m - 1}  \sum\limits_{a,b = 0}^{d-1}  \dfrac{c(ab|xy)}{\sum_{\alpha \beta} c(\alpha \beta|xy)} \left [ s_{xy}^{ab} + \dfrac{s_{x}^{a}}{m}    +  \dfrac{s_{y}^{b}}{m}   \right ].
\end{equation}\medskip

\noindent Finally, Gaussian error propagation and the Poisson statistics of the recorded coincidence count are considered to calculate $\Delta \mathcal{Q}$. The Possonian nature of the coincidence counts gives squared error  $(\Delta c(ab|xy))^2=c(ab|xy)$. The general expression for the experimental error is then
\begin{equation}\tag{S3}
    \Delta \mathcal{Q} = \sqrt{ \sum_{abxy}\left( \dfrac{\partial \mathcal{Q}}{\partial c(ab|xy)} \right)^2 c(ab|xy)},
\end{equation}
\noindent and straightforward calculation leads to 
\begin{equation}\tag{S4}
\dfrac{\partial \mathcal{Q}}{\partial c(a'b'|x'y')} = \dfrac{1 }{ \left( \sum_{ab} c(ab|x'y') \right)^2 } \left[  \left( s_{x'y'}^{a'b'}  + \dfrac{1}{m}\left( s_{x'}^{a'} + s_{y'}^{b'} \right) \right) \sum_{ab} c(ab|x'y') \right. - \left.  \sum_{ab}\left( s_{x'y'}^{ab}  + \dfrac{1}{m}\left( s_{x'}^{a} + s_{y'}^{b} \right) \right) c(ab|x'y')   \right].
\end{equation}

\begin{figure}[t]
\centering
\includegraphics[width=0.58\textwidth]{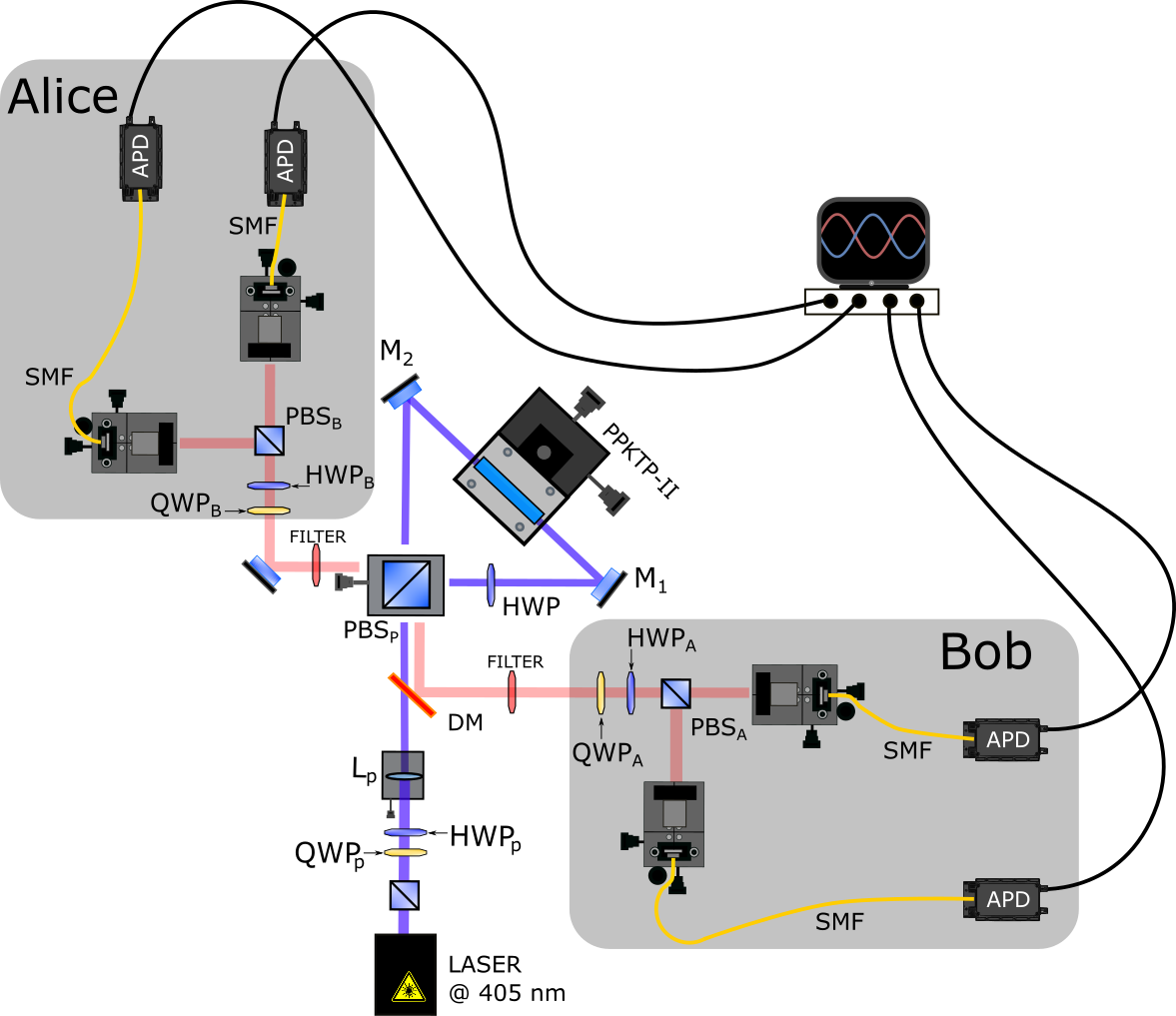}
\caption{Experimental setup used in Ref. \cite{GMMGCJAL19} to implement randomness certification and self-testing using a tunable, high-quality source of polarization-entangled down-converted photons}
\label{Fig3}
\end{figure}

\section{Experimental Details}
\label{sec:exp}
For testing our method, we use the statistics recorded in Ref. \cite{GMMGCJAL19}, where the authors aim to study randomness certification behavior and self-testing in a practical Bell scenario, considering five different partially entangled states (PES). The experiment (depicted in Fig.\ref{Fig3}) was performed using a high-purity, tunable polarization entanglement source of photons generated in the spontaneous parametric down-conversion (SPDC) process. A Sagnac interferometer, composed of two laser mirrors (M$_{1}$ and M$_{2}$), a half-wave plate (HWP), and a polarizing beam-splitter (PBS$_{p}$) cube, combined with a type-II periodically poled potassium titanyl phosphate (PPKTP) nonlinear crystal were used. The PPKTP crystal was pumped by a continuous-wave laser, operating at $405$ nm, to create degenerate down-converted photons at $810$ nm. The two propagation modes inside the interferometer (clockwise and counter-clockwise) for the generated down-converted photons overlap inside the PBS$_p$ cube, resulting in the polarized-entangled state
\begin{equation}
|\psi(\vartheta)\rangle=\cos(\vartheta) \ket{HV} + \sin(\vartheta) \ket{VH},\nonumber
\end{equation}
where the angle $\vartheta$ defines the linear polarization mode of the pump beam $\cos(\vartheta) \vert H \rangle + \sin(\vartheta) \vert V \rangle$. Therefore, the amount of entanglement, given by the concurrence $C=\sin(2\vartheta)$, can be adjusted using the half-wave (HWP$_{p}$) and the quarter-wave plate (QWP$_{p}$) located at the pump beam propagation path. To ensure the degenerate generation of the down-converted photons, Semrock high-quality narrow bandpass filters centered at $810$ nm were used, with $0.5$ nm of bandwidth and a peak transmission $>90\%$. Furthermore, to prevent distinguishability between the spatial and polarization modes, the authors couple the generated down-converted photons into single-mode optical fibers. To maximize the coincidence counts, they follow a numerical model proposed in Ref. \cite{LT05}. The optimal coupling condition is reached when $\omega_{SPDC}=\sqrt{2}\omega_{p}$, where $\omega_{p}$  and $\omega_{SPDC}$  are the waist mode of the pump beam and the down-converted photon at the center of the PPKTP crystal, respectively. These conditions were satisfied using a $20$ cm focal length for $L_{p}$ lens and $10$X objective lenses to couple the down-converted photons into the optical fibers.

The local projective measurements involved in the tilted Bell inequality (Eq. 7, main text) were implemented using the typical polarization analyzer, which consists in the HWP$_{A}$ (HWP$_{B}$), the QWP$_{A}$ (QWP$_{B}$), and the PBS$_{A}$ (PBS$_{B}$) for Alice (Bob). To reach the high overall visibility required for randomness certification and self-testing, an electronic circuit capable of implementing up to $500$ ps coincidence window was used, reducing the accidental coincidence rate probability \cite{GMGCFAL18, GGGCBD16}. PerkinElmer single-photon avalanche detectors were placed at the output mode for each PBS to record the photon statistic and estimate the set of probabilities $p(a,b|x,y)$ used for our analysis. The overall two-photon visibility obtained was $(99.7\pm 0.3)\%$ while the logical and diagonal polarization bases were measured.

\section{Canonical form of Bell inequalities}\label{ap:canonical}
To transform any bipartite Bell inequality with $m$ settings and two outcomes to its canonical form, i.e. depending on outputs $a=b=0$ only, the following identities have to be considered for local
\begin{eqnarray}
p_A(1|x)&=&1-p_A(0|x)\nonumber\\
p_B(1|y)&=&1-p_B(0|y),\nonumber
\end{eqnarray}
and joint probabilities
\begin{eqnarray}
p(0,1|x,y)&=&p_A(0|x)-p(0,0|x,y)\nonumber\\
p(1,0|x,y)&=&p_B(0|y)-p(0,0|x,y)\nonumber\\
p(1,1|x,y)&=&1-p_A(0|x)-p_B(0|y)+p(0,0|x,y)\nonumber,
\end{eqnarray}
for every $x,y=0,m-1$.


\begin{thebibliography}{99}

\bibitem{MY98} D. Mayers, A. Yao, Quantum cryptography with imperfect apparatus, \textit{Proceedings 39th Annual Symposium on Foundations of Computer Science}, pp. 503-509 (1998). 

\bibitem{ABGPS07} A. Acin, N. Brunner, N. Gisin, S. Massar, S. Pironio, V. Scarani., Device-independent security of quantum cryptography against collective attacks, \textit{Phys. Rev. Lett.} \textbf{98}, 230501 (2007). 

\bibitem{PB11} M. Pawlowski, N. Brunner, Semi-device-independent security of one-way quantum key distribution, \textit{Phys. Rev. A} \textbf{84}, 010302(R) (2011). 

\bibitem{PAM10} S. Pironio, Acín, A., S. Massar \emph{et al.}, Random numbers certified by Bell’s theorem, \textit{ Nature} \textbf{464}, 1021–1024 (2010). 

\bibitem{JP11} M. Junge, C. Palazuelos, Large violation of Bell inequalities with low entanglement, \textit{Comm. Math. Phys}. \textbf{306}, 3, 695-746 (2011). 

\bibitem{BCPSW14} Bell nonlocality, N. Brunner, D. Cavalcanti, S. Pironio, V. Scarani, S. Wehner, \textit{Rev. Mod. Phys.} \textbf{86}, 419 (2014).

\bibitem{AMP12} A. Acin, S. Massar, S. Pironio,  Randomness versus nonlocality and entanglement, \textit{Phys. Rev. Lett}. \textit{108}, 10, 100402, (2012).

\bibitem{CGLMP02} D. Collins, N. Gisin, N. Linden, S. Massar, S. Popescu, Bell inequalities for arbitrarily high dimensional systems, \textit{Phys. Rev. Lett.} \textbf{88}, 040404 (2002).

\bibitem{VW11} T. Vidick, S. Wehner, More non-locality with less entanglement, \textit{Phys. Rev. A} \textbf{83}, 052310 (2011).

\bibitem{GMMGCJAL19} S. Gómez, A. Mattar, I. Machuca, E. S. Gómez, D. Cavalcanti, O. Jiménez Farías, A. Acín, G. Lima, Experimental investigation of partially entangled states for device-independent randomness generation and self-testing protocols, \textit{Phys. Rev. A} \textbf{99}, 032108 (2019).

\bibitem{GMGCFAL18} S. Gómez, A. Mattar, E. S. Gómez, D. Cavalcanti, O. Jiménez Farías, A. Acín, and G. Lima, Experimental nonlocality-based randomness generation with nonprojective measurements, \textit{Phys. Rev. A} \textbf{97}, 040102 (2018).

\bibitem{GGGCBD16} E. S. Gómez, S. Gómez, P. González, G. Cañas, J. F. Barra, A. Delgado, G. B. Xavier, A. Cabello, M. Kleinmann, T. Vértesi et al., \textit{Phys. Rev. Lett.} \textbf{117}, 260401 (2016).

\bibitem{B64} J. Bell, On the Einstein Podolsky Rosen paradox, \textit{Physics} \textbf{1}, 195 (1964).

\bibitem{BZPZ04}C. Brukner, M. Zukowski, J. Wei Pan, A. Zeilinger, Violation of Bell's inequality: criterion for quantum communication complexity advantage,  	\textit{Phys. Rev. Lett.} \textbf{92}, 127901 (2004).

\bibitem{SRZBC18} P. Sheng Lin, D. Rosset, Y. Zhang, J.D. Bancal, Y. Cherng Liang, Device-independent point estimation from finite data and its application to device-independent property estimation, \textit{Phys. Rev. A} \textbf{97}, 032309 (2018).

\bibitem{TZB20} A. Tavakoli, M. \.{Z}ukowski, C. Brukner, Does violation of a Bell inequality always imply quantum advantage in a communication complexity problem? \textit{Quantum} \textbf{4}, 316 (2020). 

\bibitem{C44} H. Curry, The method of steepest descent for non-linear minimization problems, \textit{Quart. Appl. Math.} \textbf{2}, 3, 258--261 (1944). 

\bibitem{NM65} J. Nelder, R. Mead, A Simplex Method for Function Minimization, \textit{Computer Journal} \textbf{7}, 308--313 (1965).
 https://academic.oup.com/comjnl/article-abstract/7/4/308/354237?redirectedFrom=fulltext

\bibitem{SP97} R. Storn, K. Price, Differential evolution - a simple and efficient heuristic for global optimization over continuous spaces. \textit{Journal of Global Optimization} \textbf{11}, 341--359 (1997). 

\bibitem{KGV83} S. Kirkpatrick, C. Gelatt Jr., M. Vecchi, Optimization by Simulated Annealing, \textit{Science}  \textbf{13}, 220, 4598, 671-680 (1983).
DOI: 10.1126/science.220.4598.671

\bibitem{code} \url{https://github.com/Danuzco/nonlocality_certification}

\bibitem{LT05} D. Ljunggren, Maria Tengner, Optimal focusing for maximal collection of entangled narrow-band photon pairs into single-mode fibers, \textit{Phys. Rev. A} \textit{72}, 062301 (2005).   

\bibitem{DN2004} Daniel Collins, Nicolas Gisin, A relevant two qubit Bell inequality inequivalent to the CHSH inequality, \textit{Phys. A: Math. Gen.} \textbf{37}, 1775 (2004).

\bibitem{Pironio2010} Tamás Vértesi, Stefano Pironio, Nicolas Brunner, Closing the Detection Loophole in Bell Experiments Using Qudits, \textit{Phys. Rev. Lett.} \textbf{104}, 060401 (2010).

\end{thebibliography}
\end{document}